\documentstyle[12pt,fleqn]{article}

\font\twlgot =eufm10 scaled \magstep1 \font\egtgot =eufm8
\font\sevgot =eufm7

\font\twlmsb =msbm10 scaled \magstep1 \font\egtmsb =msbm8
\font\sevmsb =msbm7

\newfam\gotfam

\textfont\gotfam\twlgot \scriptfont\gotfam\egtgot
\scriptscriptfont\gotfam\sevgot

\newfam\msbfam
\textfont\msbfam\twlmsb \scriptfont\msbfam\egtmsb
\scriptscriptfont\msbfam\sevmsb
\def\Bbb{\protect\pBbb}
\def\pBbb{\relax\ifmmode\expandafter\Bb\else\typeout{You cann't use
Bbb in text mode}\fi}
\def\Bb #1{{\fam\msbfam\relax#1}}

\def\op#1{\mathop{{\it\fam0} #1}\limits}

\newcommand{\beq}{\begin{equation}}
\newcommand{\eeq}{\end{equation}}
\newcommand{\ben}{\begin{eqnarray}}
\newcommand{\een}{\end{eqnarray}}
\newcommand{\be}{\begin{eqnarray*}}
\newcommand{\ee}{\end{eqnarray*}}

\newcommand{\cA}{{\cal A}}

\newcommand{\cL}{{\cal L}}

\newcommand{\cE}{{\cal E}}

\newcommand{\al}{\alpha}
\newcommand{\bt}{\beta}
\newcommand{\dl}{\delta}
\newcommand{\la}{\lambda}

\newcommand{\F}{\Phi}
\newcommand{\p}{\pi}

\newcommand{\m}{\mu}

\newcommand{\si}{\sigma}
\newcommand{\Si}{\Sigma}

\newcommand{\w}{\wedge}

\newcommand{\wh}{\widehat}
\newcommand{\ol}{\overline}

\newcommand{\dr}{\partial}

\newcommand{\ot}{\otimes}

\newcounter{eqalph}
\newcounter{equationa}
\newcounter{theorem}
\newcounter{remark}
\newcounter{example}
\newcounter{proposition}
\newcounter{lemma}
\newcounter{corollary}
\newcounter{definition}

\def\theremark{\arabic{remark}}

\def\thedefinition{\arabic{definition}}

\newenvironment{proof}{\noindent
{\bf Proof.}}{\hfill $\Box$ \medskip}

\newenvironment{ex}{\refstepcounter{remark}\medskip\noindent{\bf Example
\theremark.}}{\medskip}
\newenvironment{theo}{\refstepcounter{definition} \medskip
\noindent{\bf Theorem \thedefinition.}}{\medskip}

\newenvironment{lem}{\refstepcounter{definition} \medskip
\noindent{\bf Lemma \thedefinition.}}{\medskip}

\newcommand{\mar}[1]{}

\begin{document}

\hbox{}

\begin{center}

{\large \bf Lagrangian dynamics of submanifolds. Relativistic
mechanics}

\bigskip
\bigskip

{\sc G. Sardanashvily}
\bigskip

{\it Department of Theoretical Physics, Moscow State University,
Moscow, Russia}

\end{center}

\bigskip
\bigskip

\begin{small}

\noindent Geometric formulation of Lagrangian relativistic
mechanics in the terms of jets of one-dimensional submanifolds is
generalized to Lagrangian theory of submanifolds of arbitrary
dimension.


\end{small}


\bigskip
\bigskip

\section{Introduction}

Classical non-relativistic mechanics is adequately formulated as
Lagrangian and Hamiltonian theory on a fibre bundle $Q\to\Bbb R$
over the time axis $\Bbb R$, where $\Bbb R$ is provided with the
Cartesian coordinate $t$ possessing the transition functions
$t'=t+$const. \cite{eche,book10,leon,book98,sard98}. A velocity
space of non-relativistic mechanics is the first order jet
manifold $J^1Q$ of sections of $Q\to \Bbb R$. Lagrangians of
non-relativistic mechanics are defined as densities on $J^1Q$.
This formulation is extended to time-reparametrized
non-relativistic mechanics subject to time-dependent
transformations which are bundle automorphisms of $Q\to\Bbb R$
\cite{book10,book98}.

Thus, one can think of non-relativistic mechanics as being
particular classical field theory on fibre bundles over $X=\Bbb
R$. However, an essential difference between non-relativistic
mechanics and field theory on fibre bundles $Y\to X$, $\dim X>1$,
lies in the fact that connections on $Q\to\Bbb R$ always are flat.
Therefore, they fail to be dynamic variables, but characterize
non-relativistic reference frames.

In comparison with non-relativistic mechanics, relativistic
mechanics admits transformations of the time depending on other
variables, e.g., the Lorentz transformations in Special Relativity
on a Minkowski space $Q=\Bbb R^4$. Therefore, a configuration
space $Q$ of relativistic mechanics has no preferable fibration
$Q\to\Bbb R$, and its velocity space is the first order jet
manifold $J^1_1Q$ of one-dimensional submanifolds of a
configuration space $Q$ \cite{book10,book98,sard10}. Fibres of the
jet bundle $J^1_1Q\to Q$ are projective spaces, and one can think
of them as being spaces of the three-velocities of a relativistic
system. The four-velocities of a relativistic system are
represented by elements of the tangent bundle $TQ$ of a
configuration space $Q$.

This work is devoted to generalization of the above mentioned
formulation of relativistic mechanics to the case of submanifolds
of arbitrary dimension.

Let us consider $n$-dimensional submanifolds of an $m$-dimensional
smooth real manifold $Z$. The notion of jets of submanifolds
\cite{book09,kras,modu} generalizes that of jets of sections of
fibre bundles, which are particular jets of submanifolds (Section
2). Namely, a space of jets of submanifolds admits a cover by
charts of jets of sections. Just as in relativistic mechanics, we
restrict our consideration to first order jets of submanifolds
which form a smooth manifold $J^1_nZ$. One however meets a problem
how to develop Lagrangian formalism on a manifold $J^1_nZ$ because
it is not a fibre bundle.

For this purpose, we associate to $n$-dimensional submanifolds of
$Z$ the sections of a trivial fibre bundle
\mar{str}\beq
\pi:Z_\Si=\Si\times Z\to \Si, \label{str}
\eeq
where $\Si$ is some $n$-dimensional manifold. We obtain a relation
between the elements of $J^1_nZ$ and the jets of sections of the
fibre bundle (\ref{str})(Section 3). This relation fails to be
one-to-one correspondence. The ambiguity contains, e.g.,
diffeomorphisms of $\Si$. Then Lagrangian formalism on a fibre
bundle $Z_\Si\to \Si$ is developed in a standard way, but a
Lagrangian is required to possess the gauge symmetry (\ref{s59})
which leads to the rather restrictive Noether identities
(\ref{s60}) (Section 4).

If $n=2$, this is the case, e.g., of the Nambu--Goto Lagrangian
(\ref{s140}) of classical string theory (Example \ref{sss}).

If $n=1$, solving these Noether identities, we obtain a generic
Lagrangian (\ref{kk1}) of relativistic mechanics (Section 5).

These examples confirm the correctness of our description of
Lagrangian dynamics of submanifolds of a manifold $Z$ as that of
sections of the fibre bundle $Z_\Si$ (\ref{str}).

\section{Jets of submanifolds}

Given an $m$-dimensional smooth real manifold $Z$, a $k$-order jet
of $n$-dimensional submanifolds of $Z$ at a point $z\in Z$ is
defined as an equivalence class $j^k_zS$ of $n$-dimensional
imbedded submanifolds of $Z$ through $z$ which are tangent to each
other at $z$ with order $k>0$. Namely, two submanifolds
\be
i_S: S\to Z,\qquad  i_{S'}: S'\to Z
\ee
through a point $z\in Z$ belong to the same equivalence class
$j^k_zS$ if and only if the images of the $k$-tangent morphisms
\be
T^ki_S: T^k_zS\to T^k_zZ, \qquad  T^ki_{S'}: T^k_zS'\to T^k_zZ
\ee
coincide with each other. The set
\be
J^k_nZ=\op\bigcup_{z\in Z} j^k_zS
\ee
of $k$-order jets of submanifolds is a finite-dimensional real
smooth manifold, called the $k$-order jet manifold of submanifolds
\cite{book09,kras,modu}.

Let $Y\to X$ be an $m$-dimensional fibre bundle over an
$n$-dimensional base $X$ and $J^kY$ the $k$-order jet manifold of
sections of $Y\to X$. Given an imbedding $\Phi:Y\to Z$, there is
the natural injection
\be
J^k\Phi: J^kY\to J^k_nZ, \qquad j^k_xs \to [\Phi\circ
s]^k_{\Phi(s(x))},
\ee
where $s$ are sections of $Y\to X$. This injection defines a chart
on $J^k_nZ$. These charts provide a manifold atlas of $J^k_nZ$.

Let us restrict our consideration to first order jets of
submanifolds. There is obvious one-to-one correspondence
\mar{s10}\beq
\zeta: j^1_zS \to V_{j^1_zS}\subset T_zZ  \label{s10}
\eeq
between the jets $j^1_zS$ at a point $z\in Z$ and the
$n$-dimensional vector subspaces of the tangent space $T_zZ$ of
$Z$ at $z$. It follows that $J^1_nZ$ is a fibre bundle
\mar{s3}\beq
\rho:J^1_nZ\to Z \label{s3}
\eeq
with the structure group $GL(n,m-n;\Bbb R)$ of linear
transformations of a vector space $\Bbb R^m$ which preserve its
subspace $\Bbb R^n$. The typical fibre of the fibre bundle
(\ref{s3}) is a Grassmann manifold $GL(m;\Bbb R)/GL(n,m-n;\Bbb
R)$. This fibre bundle is endowed with the following coordinate
atlas.

Let $\{(U;z^A)\}$ be a coordinate atlas of $Z$. Let us provide $Z$
with an atlas obtained by replacing every chart $(U,z^A)$ of $Z$
with the
\be
{m\choose n}=\frac{m!}{n!(m-n)!}
\ee
charts on $U$ which correspond to different partitions of $(z^A)$
in collections of $n$ and $m-n$ coordinates
\mar{5.8}\beq
(U; x^a,y^i), \qquad a=1,\ldots,n,  \qquad
i=1,\ldots,m-n.\label{5.8}
\eeq
The transition functions between the coordinate charts (\ref{5.8})
associated with a coordinate chart $(U,z^A)$ are reduced to
exchange between coordinates $x^a$ and $y^i$. Transition functions
between arbitrary coordinate charts (\ref{5.8}) take the form
\mar{5.26} \beq
x'^a = x'^a (x^b, y^k), \qquad y'^i = y'^i (x^b, y^k).
\label{5.26}
\eeq
Let $J^0_nZ$ denote a manifold $Z$ provided with the coordinate
atlas (\ref{5.8})  -- (\ref{5.26}).

Given this atlas of $J^0_nZ=Z$, the first order jet manifold
$J^1_nZ$ is endowed with the coordinate charts
\mar{5.31}\beq
(\rho^{-1}(U)=U\times\Bbb R^{(m-n)n}; x^a,y^i,y^i_a), \label{5.31}
\eeq
possessing the following transition functions. With respect to the
coordinates (\ref{5.31}) on the jet manifold $J^1_nZ$ and the
induced fibre coordinates $(\dot x^a, \dot y^i)$ on the tangent
bundle $TZ$, the above mentioned correspondence $\zeta$
(\ref{s10}) reads
\be
\zeta: (y^i_a) \to \dot x^a(\dr_a +y^i_a(j^1_zS)\dr_i).
\ee
It implies the relations
\mar{s0,1}\ben
&&  y'^j_a= (\frac{\dr y'^j}{\dr y^k} y^k_b + \frac{\dr y'^j}{\dr x^b})
(\frac{\dr x^b}{\dr y'^i}y'^i_a + \frac{\dr x^b}{\dr x'^a}), \label{s0}\\
&& (\frac{\dr x^b}{\dr y'^i}y'^i_a + \frac{\dr x^b}{\dr x'^a})
(\frac{\dr x'^c}{\dr y^k} y^k_b + \frac{\dr x'^c}{\dr
x^b})=\dl^c_a,\label{s1}
\een
which jet coordinates $y^i_a$ must satisfy under coordinate
transformations (\ref{5.26}). Let us consider a non-degenerate
$n\times n$ matrix $M$ with the entries
\be
 M^c_b=(\frac{\dr x'^c}{\dr
y^k}y^k_b + \frac{\dr x'^c}{\dr x^b}).
\ee
Then the relations (\ref{s1}) lead to the equalities
\be
(\frac{\dr x^b}{\dr y'^i} y'^i_a + \frac{\dr x^b}{\dr x'^a})=
(M^{-1})^b_a.
\ee
Hence, we obtain the transformation law of first order jet
coordinates
\mar{s2}\beq
 y'^j_a=
( \frac{\dr y'^j}{\dr y^k} y^k_b+ \frac{\dr y'^j}{\dr x^b})
(M^{-1})^b_a. \label{s2}
\eeq
In particular, if coordinate transition functions $x'^a$
(\ref{5.26}) are independent of coordinates $y^k$, the
transformation law (\ref{s2}) comes to the familiar
transformations of jets of sections.

\section{The fibre bundle $Z_\Si$}

Given a coordinate chart (\ref{5.31}) of $J^1_nZ$, one can regard
$\rho^{-1}(U)\subset J^1_nZ$ as the first order jet manifold
$J^1U$ of sections of a fibre bundle
\be
\chi:U\ni (x^a,y^i)\to (x^a)\in \chi(U).
\ee
The graded differential algebra of exterior forms on
$\rho^{-1}(U)$ is generated by horizontal forms $dx^a$ and contact
forms $dy^i-y^i_adx^a$. Coordinate transformations (\ref{5.26})
and (\ref{s2}) preserve the ideal of contact forms, but horizontal
forms are not transformed into horizontal forms, unless coordinate
transition functions $x'^a$ (\ref{5.26}) are independent of
coordinates $y^k$. Therefore, one can develop first order
Lagrangian formalism with a Lagrangian $L=\cL d^nx$ on a
coordinate chart $\rho^{-1}(U)$, but this Lagrangian fails to be
globally defined on $J^1_nZ$.

In order to overcome this difficulty, let us consider the trivial
fibre bundle $Z_\Si\to \Si$ (\ref{str}) whose trivialization
throughout holds fixed. This fibre bundle is provided with an
atlas of coordinate charts
\mar{s20'}\beq
(U_\Si\times U; \si^\m,x^a,y^i), \label{s20'}
\eeq
where $(U; x^a,y^i)$ are the above mentioned coordinate charts
$(\ref{5.8})$ of a manifold $J^0_n Z$. The coordinate charts
(\ref{s20'}) possess transition functions
\mar{s21}\beq
\si'^\m=\si^\m(\si^\nu), \qquad x'^a = x'^a (x^b, y^k), \qquad
y'^i = y'^i (x^b, y^k). \label{s21}
\eeq
Let $J^1Z_\Si$ be the first order jet manifold of the fibre bundle
(\ref{str}). Since the trivialization (\ref{str}) is fixed, it is
a vector bundle
\be
\pi^1:J^1Z_\Si\to Z_\Si
\ee
isomorphic to the tensor product
\mar{s30}\beq
J^1Z_\Si= T^*\Si\op\ot_{\Si\times Z} TZ \label{s30}
\eeq
of the cotangent bundle $T^*\Si$ of $\Si$ and the tangent bundle
$TZ$ of $Z$ over $Z_\Si$.

Given the coordinate atlas (\ref{s20'}) - (\ref{s21}) of $Z_\Si$,
the jet manifold  $J^1Z_\Si$ is endowed with the coordinate charts
\mar{s14}\beq
 ((\pi^1)^{-1}(U_\Si\times U)=U_\Si\times U\times\Bbb R^{mn};
\si^\m,x^a,y^i,x^a_\m, y^i_\m), \label{s14}
\eeq
possessing transition functions
\mar{s16'}\beq
x'^a_\m=(\frac{\dr x'^a}{\dr y^k}y^k_\nu + \frac{\dr x'^a}{\dr
x^b}x^b_\nu )\frac{\dr \si^\nu}{\dr \si'^\m}, \quad
y'^i_\m=(\frac{\dr y'^i}{\dr y^k}y^k_\nu + \frac{\dr y'^i}{\dr
x^b}x^b_\nu)\frac{\dr \si^\nu}{\dr \si'^\m}. \label{s16'}
\eeq
Relative to the coordinates (\ref{s14}), the bundle isomorphism
(\ref{s30}) takes the form
\be
(x^a_\m, y^i_\m) \to d\si^\m\ot(x^a_\m \dr_a + y^i_\m \dr_i).
\ee

Obviously, a jet $(\si^\m,x^a,y^i,x^a_\m, y^i_\m)$ of sections of
the fibre bundle (\ref{str}) defines some jet of $n$-dimensional
submanifolds of a manifold $\{\si\}\times Z$ through a point
$(x^a,y^i)\in Z$ if an $m\times n$ matrix with the entries
$(x^a_\m, y^i_\m)$ is of maximal rank $n$. This property is
preserved under the coordinate transformations (\ref{s16'}). An
element of $J^1Z_\Si$ is called regular if it possesses this
property. Regular elements constitute an open subbundle of the jet
bundle $J^1Z_\Si\to Z_\Si$.

Since regular elements of $J^1Z_\Si$ characterize first jets of
$n$-dimensional submanifolds of $Z$, one hopes to describe the
dynamics of these submanifolds of a manifold $Z$ as that of
sections of the fibre bundle (\ref{str}).  For this purpose, let
us refine the relation between elements of the jet manifolds
$J^1_nZ$ and $J^1Z_\Si$.

Let us consider the manifold product $\Si\times J^1_nZ$. It is a
fibre bundle over $Z_\Si$.  Given the coordinate atlas
(\ref{s20'}) - (\ref{s21}) of $Z_\Si$, this product is endowed
with the coordinate charts
\mar{s13}\beq
(U_\Si\times \rho^{-1}(U)=U_\Si\times U\times\Bbb R^{(m-n)n};
\si^\m,x^a,y^i, y^i_a), \label{s13}
\eeq
possessing the transition functions (\ref{s2}). Let us assign to
an element $(\si^\m,x^a,y^i, y^i_a)$ of the chart (\ref{s13}) the
elements $(\si^\m,x^a,y^i,x^a_\m, y^i_\m)$ of the chart
(\ref{s14}) whose coordinates obey the relations
\mar{s17}\beq
y^i_a x^a_\m = y^i_\m. \label{s17}
\eeq
These elements make up an $n^2$-dimensional vector space. The
relations (\ref{s17}) are maintained under  the coordinate
transformations (\ref{s21}) and the induced transformations of the
charts (\ref{s14}) and (\ref{s13}) as follows:
\be
&& y'^i_a x'^a_\m =
( \frac{\dr y'^i}{\dr y^k} y^k_c+ \frac{\dr y'^i}{\dr x^c})
(M^{-1})^c_a (\frac{\dr x'^a}{\dr y^k}y^k_\nu + \frac{\dr
x'^a}{\dr x^b}x^b_\nu)\frac{\dr \si^\nu}{\dr \si'^\m}
= \\
&& \qquad ( \frac{\dr y'^i}{\dr y^k} y^k_c+
\frac{\dr y'^i}{\dr x^c}) (M^{-1})^c_a (\frac{\dr x'^a}{\dr
y^k}y^k_b + \frac{\dr x'^a}{\dr x^b} )x^b_\nu\frac{\dr
\si^\nu}{\dr \si'^\m}
=\\
&& \qquad (\frac{\dr y'^i}{\dr
y^k}y^k_b + \frac{\dr y'^i}{\dr x^b})x^b_\nu\frac{\dr \si^\nu}{\dr
\si'^\m}= (\frac{\dr y'^i}{\dr y^k}y^k_\nu + \frac{\dr y'^i}{\dr
x^b}x^b_\nu)\frac{\dr \si^\nu}{\dr \si'^\m}= y'^i_\m.
\ee
Thus, one can associate:
\be
 \zeta': (\si^\m,x^a,y^i, y^i_a) \to
\{(\si^\m,x^a,y^i,x^a_\m, y^i_\m) \, | \, y^i_a x^a_\m = y^i_\m\},
\ee
to each element of a manifold $\Si\times J^1_nZ$ an
$n^2$-dimensional vector space in a jet manifold $J^1Z_\Si$. This
is a subspace of elements
\be
x^a_\m d\si^\m\ot(\dr_a + y^i_a\dr_i)
\ee
of a fibre of the tensor bundle (\ref{s30}) at a point
$(\si^\m,x^a,y^i)$. This subspace always contains regular
elements, e.g., whose coordinates $x^a_\m$ form a non-degenerate
$n\times n$ matrix.

Conversely, given a regular element $j^1_zs$ of $J^1Z_\Si$, there
is a coordinate chart (\ref{s14}) such that coordinates $x^a_\m$
of $j^1_zs$  constitute a non-degenerate matrix, and $j^1_zs$
defines a unique element of $\Si\times J^1_nZ$ by the relations
\mar{s31}\beq
y^i_a=y^i_\m(x^{-1})^\m_a. \label{s31}
\eeq

Thus, we have shown the following. Let $(\si^\m,z^A)$ further be
arbitrary coordinates on the product $Z_\Si$ (\ref{str}) and
$(\si^\m,z^A,z^A_\m)$ the corresponding coordinates on the jet
manifold $J^1Z_\Si$. In these coordinates, an element of
$J^1Z_\Si$ is regular if an $m\times n$ matrix with the entries
$z^A_\m$ is of maximal rank $n$.

\begin{theo} \label{s50} \mar{s50}
(i) Any jet of submanifolds through a point $z\in Z$ defines some
(but not unique) jet of sections of a fibre bundle $Z_\Si$
(\ref{str}) through a point $\si\times z$ for any $\si\in \Si$ in
accordance with the relations (\ref{s17}).

(ii)  Any regular element of $J^1Z_\Si$ defines a unique element
of a jet manifold $J^1_nZ$ by means of the relations (\ref{s31}).
However, non-regular elements of $J^1Z_\Si$ can correspond to
different jets of submanifolds.

(iii) Two elements $(\si^\m,z^A,z^A_\m)$ and
$(\si^\m,z^A,z'^A_\m)$ of $J^1Z_\Si$ correspond to the same jet of
submanifolds if
\be
z'^A_\m= M^\nu_\mu z^A_\nu,
\ee
where $M$ is some matrix, e.g., it comes from a diffeomorphism of
$\Si$.
\end{theo}

\section{Lagrangian formalism}

Based on Theorem \ref{s50}, we can describe the dynamics of
$n$-dimensional submanifolds of a manifold $Z$ as that of sections
of the fibre bundle $Z_\Si$ (\ref{str}) for some $n$-dimensional
manifold $\Si$.

Let
\mar{zzz}\beq
L=\cL(z^A, z^A_\m) d^n\si, \label{zzz}
\eeq
be a first order Lagrangian on a jet manifold $J^1Z_\Si$. The
corresponding Euler--Lagrange operator reads
\mar{kkk}\beq
\dl L= \cE_A dz^A\w d^n\si, \qquad \cE_A= \dr_A\cL - d_\m
\dr_A^\m\cL. \label{kkk}
\eeq
It yields the Euler--Lagrange equations
\be
\cE_A= \dr_A\cL - d_\m \dr_A^\m\cL =0.
\ee

In view of Theorem \ref{s50}, it seems reasonable to require that,
in order to describe jets of $n$-dimensional submanifolds of $Z$,
the Lagrangian $L$ (\ref{zzz}) on $J^1Z_\Si$ must be invariant
under diffeomorphisms of a manifold $\Si$. To formulate this
condition, it is sufficient to consider infinitesimal generators
of one-parameter subgroups of these diffeomorphisms which are
vector fields $u=u^\m\dr_\m$ on $\Si$. Since $Z_\Si\to \Si$ is a
trivial bundle, such a vector field gives rise to a vector field
$u=u^\m\dr_\m$ on $Z_\Si$. Its jet prolongation onto $J^1Z_\Si$
reads
\mar{s59}\ben
&& J^1u= u^\m \dr_\m - z^A_\nu\dr_\m u^\nu \dr_A^\m= u^\m d_\m +[-
u^\nu z^A_\nu\dr_A - d_\m (u^\nu z^A_\nu)\dr_A^\m], \label{s59}\\
&& d_\m=\dr_\m +z^A_\m\dr_A + z^A_{\m\nu}\dr^\nu_\m+\cdots.
\nonumber
\een
One can regard it as a generalized vector field on $J^1Z_\Si$
depending on parameter functions $u^\m(\si^\nu)$, i.e., it is a
gauge transformation \cite{jmp09,book09}. Let us require that
$J^1u$ (\ref{s59}) or, equivalently, its vertical part
\be
u_V= - u^\nu z^A_\nu\dr_A - d_\m(u^\nu z^A_\nu)\dr_A^\m.
\ee
is a variational symmetry of the Lagrangian $L$ (\ref{zzz}). Then
by virtue of the second Noether theorem, the Euler--Lagrange
operator $\dl L$ (\ref{kkk}) obeys the irreducible Noether
identities
\mar{s60}\beq
z^A_\nu\cE_A=0. \label{s60}
\eeq

One can think of these identities as being a condition which the
Lagrangian $L$ on $J^1Z_\Si$ must satisfy in order to be a
Lagrangian of submanifolds of $Z$. It is readily observed that
this condition is rather restrictive.

\begin{ex} \mar{sss} \label{sss}
Let $Z$ be a locally affine manifold, i.e., a toroidal cylinder
$\Bbb R^{m-k}\times T^k$. Its tangent bundle can be provided with
a constant non-degenerate fibre metric $\eta_{AB}$. Let $\Si$ be a
two-dimensional manifold. Let us consider the $2\times 2$ matrix
with the entries
\be
h_{\m\nu}=\eta_{AB} z^A_\m z^B_\nu.
\ee
Then its determinant provides a Lagrangian
\mar{s140}\beq
L=(\det h)^{1/2} d^2\si =([\eta_{AB} z^A_1 z^B_1] [\eta_{AB} z^A_2
z^B_2]- [\eta_{AB} z^A_1 z^B_2]^2 )^{1/2} d^2\si  \label{s140}
\eeq
on the jet manifold  $J^1Z_\Si$ (\ref{s30}). This is the well
known Nambu--Goto Lagrangian of classical string theory
\cite{polch}. It satisfies the Noether identities (\ref{s60}).
\end{ex}

\section{Relativistic mechanics}

As was mentioned above, if $n=1$, we are in the case of
relativistic mechanics. In this case, one can obtain a complete
solution of the Noether identities (\ref{s60}) which provides a
generic Lagrangian of relativistic mechanics.

Given an $m$-dimensional manifold $Q$ coordinated by $(q^\la)$,
let us consider the jet manifold $J^1_1Q$ of its one-dimensional
submanifolds. It is treated as a velocity space of relativistic
mechanics \cite{book10,book98,sard10}. Let us provide $Q=J^0_1Q$
with the coordinates (\ref{5.8}):
\mar{0303}\beq
(U;x^0=q^0, y^i=q^i)= (U;q^\la). \label{0303}
\eeq
Then the jet manifold $\rho:J^1_1Q\to Q$ is endowed with the
coordinates (\ref{5.31}):
\mar{0300}\beq
(\rho^{-1}(U);q^0,q^i,q^i_0) \label{0300}
\eeq
possessing transition functions (\ref{5.26}), (\ref{s2}):
\mar{s120,'}\ben
&&q'^0=q'^0(q^0,q^k), \qquad q'^0=q'^0(q^0,q^k), \label{s120}\\
&&q'^i_0= (\frac{\dr q'^i}{\dr q^j} q^j_0 + \frac{\dr q'^i}{\dr
q^0} ) (\frac{\dr q'^0}{\dr q^j} q^j_0 + \frac{\dr q'^0}{\dr q^0}
)^{-1}. \label{s120'}
\een
A glance at the transformation law (\ref{s120'}) shows that
$J^1_1Q\to Q$ is a fibre bundle in projective spaces.

\begin{ex} \label{0310} \mar{0310}
Let $Q=M^4=\Bbb R^4$ be a Minkowski space whose Cartesian
coordinates $(q^\la)$, $\la=0,1,2,3,$ are subject to the Lorentz
transformations (\ref{s120}):
\mar{s122}\beq
q'^0= q^0{\rm ch}\al - q^1{\rm sh}\al, \quad q'^1= -q^0{\rm sh}\al
+ q^1{\rm ch}\al, \quad q'^{2,3} = q^{2,3}. \label{s122}
\eeq
Then $q'^i$ (\ref{s120'}) are exactly the Lorentz transformations
\be
q'^1_0=\frac{ q^1_0{\rm ch}\al -{\rm sh}\al}{ - q^1_0{\rm sh}\al+
{\rm ch}\al} \qquad q'^{2,3}_0=\frac{q^{2,3}_0}{ - q^1_0{\rm
sh}\al + {\rm ch}\al}
\ee
of three-velocities in Special Relativity.
\end{ex}

In view of Example \ref{0310}, let us call a velocity space
$J^1_1Q$ of relativistic mechanics the space of three-velocities,
though a dimension of $Q$ need not equal $3+1$.

In order to develop Lagrangian formalism of relativistic
mechanics, let us consider the trivial fibre bundle (\ref{str}):
\mar{str'}\beq
Q_R=\Bbb R\times Q\to \Bbb R, \label{str'}
\eeq
whose base $\Si=\Bbb R$ is endowed with a global Cartesian
coordinate $\tau$. This fibre bundle is provided with an atlas of
coordinate charts
\mar{s20a}\beq
(\Bbb R\times U; \tau,q^\la), \label{s20a}
\eeq
where $(U; q^0,q^i)$ are the coordinate charts (\ref{0303}) of a
manifold $J^0_1Q$. The coordinate charts (\ref{s20a}) possess the
transition functions (\ref{s120}). Let $J^1Q_R$ be the first order
jet manifold of the fibre bundle (\ref{str'}). Since the
trivialization (\ref{str'}) is fixed, there is the canonical
isomorphism (\ref{s30}) of $J^1Q_R$ to the vertical tangent bundle
\mar{s30'}\beq
J^1Q_R= VQ_R= \Bbb R\times TQ \label{s30'}
\eeq
of $Q_R\to \Bbb R$.

Given the coordinate atlas (\ref{s20a}) of $Q_R$, a jet manifold
$J^1Q_R$ is endowed with the coordinate charts
\mar{s14'}\beq
 ((\pi^1)^{-1}(\Bbb R\times U)=\Bbb R\times U\times\Bbb R^m;
\tau,q^\la,q^\la_\tau), \label{s14'}
\eeq
possessing transition functions
\mar{s16a}\beq
q'^\la_\tau=\frac{\dr q'^\la}{\dr q^\m}q^\m_\tau. \label{s16a}
\eeq
Relative to the coordinates (\ref{s14'}), the isomorphism
(\ref{s30'}) takes the form
\mar{0305}\beq
(\tau,q^\m,q^\m_\tau) \to (\tau,q^\m,\dot q^\m=q^\m_\tau).
\label{0305}
\eeq

\begin{ex} \label{0311} \mar{0311} Let $Q=M^4$ be a Minkowski
space in Example \ref{0310} whose Cartesian coordinates
$(q^0,q^i)$ are subject to the Lorentz transformations
(\ref{s122}). Then the corresponding transformations (\ref{s16a})
take the form
\be
q'^0_\tau= q^0_\tau{\rm ch}\al - q^1_\tau{\rm sh}\al, \quad
q'^1_\tau= -q^0_\tau{\rm sh}\al + q^1_\tau{\rm ch}\al, \quad
q'^{2,3}_\tau = q^{2,3}_\tau
\ee
of transformations of four-velocities in Special Relativity.
\end{ex}

In view of Example \ref{0311}, we agree to call fibre elements of
$J^1Q_R\to Q_R$ the four-velocities though the dimension of $Q$
need not equal 4. Due to the canonical isomorphism $q^\la_\tau\to
\dot q^\la$ (\ref{s30'}), by four-velocities also are meant the
elements of the tangent bundle $TQ$, which is called the space of
four-velocities.

In accordance with the terminology of Section 3, the non-zero jet
(\ref{0305}) of sections of the fibre bundle (\ref{str'}) is
regular, and it defines some jet of one-dimensional submanifolds
of a manifold $\{\tau\}\times Q$ through a point $(q^0,q^i)\in Q$.
Although this is not one-to-one correspondence, just as in Section
4, one can describe the dynamics of one-dimensional submanifolds
of a manifold $Q$ as that of sections of the fibre bundle
(\ref{str'}).

Let us consider the manifold product $\Bbb R\times J^1_1Q$. It is
a fibre bundle over $Q_R$.  Given the coordinate atlas
(\ref{s20a}) of $Q_R$, this product is endowed with the coordinate
charts (\ref{s13}):
\mar{s13'}\beq
(U_R\times \rho^{-1}(U)=U_R\times U\times\Bbb R^{m-1};
\tau,q^0,q^i, q^i_0), \label{s13'}
\eeq
possessing transition functions (\ref{s120}) -- (\ref{s120'}). Let
us assign to an element $(\tau,q^0,q^i, q^i_0)$ of the chart
(\ref{s13'}) the elements $(\tau,q^0,q^i,q^0_\tau, q^i_\tau)$ of
the chart (\ref{s14'}) whose coordinates obey the relations
(\ref{s17}):
\mar{s17'}\beq
q^i_0 q^0_\tau = q^i_\tau. \label{s17'}
\eeq
These elements make up a one-dimensional vector space. The
relations (\ref{s17'}) are maintained under  coordinate
transformations (\ref{s120'}) and (\ref{s16a}). Thus, one can
associate to each element of the manifold $\Bbb R\times J^1_1Q$ a
one-dimensional vector space
\mar{s25}\beq
(\tau,q^0,q^i, q^i_0) \to \{(\tau,q^0,q^i,q^0_\tau, q^i_\tau) \, |
\, q^i_0 q^0_\tau = q^i_\tau\}, \label{s25}
\eeq
in a jet manifold $J^1Q_R$. This is a subspace of elements
$q^0_\tau (\dr_0 + q^i_0\dr_i)$ of a fibre of the vertical tangent
bundle (\ref{s30'}) at a point $(\tau,q^0,q^i)$. Conversely, given
a non-zero element (\ref{0305}) of $J^1Q_R$, there is a coordinate
chart (\ref{s14'}) such that this element defines a unique element
of $\Bbb R\times J^1_1Q$ by the relations (\ref{s31}):
\mar{s31'}\beq
q^i_0=\frac{q^i_\tau}{q^0_\tau}. \label{s31'}
\eeq

Thus, we come to Theorem \ref{s50} for the case $n=1$ as follows.
Let $(\tau,q^\la)$ further be arbitrary coordinates on the product
$Q_R$ (\ref{str'}) and $(\tau,q^\la,q^\la_\tau)$ the corresponding
coordinates on a jet manifold $J^1Q_R$.

\begin{theo} \label{s50'} \mar{s50'}
(i) Any jet of submanifolds through a point $q\in Q$ defines some
(but not unique) jet of sections of the fibre bundle $Q_R$
(\ref{str'}) through a point $\tau\times q$ for any $\tau\in \Bbb
R$ in accordance with the relations (\ref{s17'}).

(ii)  Any non-zero element of $J^1Q_R$ defines a unique element of
the jet manifold $J^1_1Q$ by means of the relations (\ref{s31'}).
However, non-zero elements of $J^1Q_R$ can correspond to different
jets of submanifolds.

(iii) Two elements $(\tau,q^\la,q^\la_\tau)$ and
$(\tau,q^\la,q'^\la_\tau)$ of $J^1Q_R$ correspond to the same jet
of submanifolds if $q'^\la_\tau=r q^\la_\tau$, $r\in\Bbb
R\setminus \{0\}$.
\end{theo}

In the case of a Minkowski space $Q=M^4$ in Examples \ref{0310}
and \ref{0311}, the equalities (\ref{s17'}) and (\ref{s31'}) are
the familiar relations between three- and four-velocities.

Let
\mar{s40}\beq
L=\cL(\tau,q^\la, q^\la_\tau) d\tau, \label{s40}
\eeq
be a first order Lagrangian on a jet manifold $J^1Q_R$. The
corresponding Lagrange operator reads
\mar{s41}\beq
\dl L= \cE_\la dq^\la\w d\tau, \qquad \cE_\la= \dr_\la\cL - d_\tau
\dr_\la^\tau\cL. \label{s41}
\eeq
Let us require that, in order to describe jets of one-dimensional
submanifolds of $Q$, the Lagrangian $L$ (\ref{s40}) on $J^1Q_R$
possesses a gauge symmetry given by vector fields
$u=u(\tau)\dr_\tau$ on $Q_R$ or, equivalently, their vertical part
\mar{kk33}\beq
u_V= - u(\tau) q^\la_\tau\dr_\la, \label{kk33}
\eeq
which are generalized vector fields on $Q_R$. Then the variational
derivatives of this Lagrangian obey the Noether identities
(\ref{s60}):
\mar{s60'}\beq
q^\la_\tau\cE_\la=0. \label{s60'}
\eeq
We call such a Lagrangian the relativistic Lagrangian.

In order to obtain a generic form of a relativistic Lagrangian
$L$, let us regard the Noether identities (\ref{s60'}) as an
equation for $L$. It admits the following solution. Let
\be
\frac1{2N!}G_{\al_1\ldots\al_{2N}}(q^\nu)dq^{\al_1}\vee\cdots\vee
dq^{\al_{2N}}
\ee
be a symmetric tensor field on $Q$ such that a function
\mar{kk6}\beq
G=G_{\al_1\ldots\al_{2N}}(q^\nu)\dot q^{\al_1}\cdots\dot
q^{\al_{2N}} \label{kk6}
\eeq
is positive ($G>0$) everywhere on $TQ\setminus \wh 0(Q)$, where
$\wh 0(Q)$ is the global zero section of $TQ\to Q$. Let
$A=A_\m(q^\nu) dq^\m$ be a one-form on $Q$. Given the pull-back of
$G$ and $A$ onto $J^1Q_R$ due to the canonical isomorphism
(\ref{s30'}), we define a Lagrangian
\mar{kk1}\beq
L=(G^{1/2N} + q^\m_\tau A_\m)d\tau, \qquad
G=G_{\al_1\ldots\al_{2N}}q^{\al_1}_\tau\cdots q^{\al_{2N}}_\tau,
\label{kk1}
\eeq
on $J^1Q_R\setminus (\Bbb R\times \wh 0(Q))$. The corresponding
Lagrange equations read
\mar{kk2,3}\ben
&&\cE_\la = \frac{\dr_\la G}{2NG^{1-1/2N}} - d_\tau\left(\frac{\dr_\la^\tau
G}{2NG^{1-1/2N}}\right)+ F_{\la\m}q^\m_\tau= \label{kk2}\\
&& \qquad E_\bt[\dl^\bt_\la -
 q^\bt_\tau G_{\la\nu_2\ldots\nu_{2N}}q^{\nu_2}_\tau\cdots
q^{\nu_{2N}}_\tau G^{-1}]G^{1/2N-1}=0, \nonumber\\
&& E_\bt= \left(\frac{\dr_\bt
G_{\m\al_2\ldots\al_{2N}}}{2N}- \dr_\m
G_{\bt\al_2\ldots\al_{2N}}\right) q^\m_\tau q^{\al_2}_\tau\cdots
q^{\al_{2N}}_\tau - \label{kk3}\\
&& \qquad (2N-1)G_{\bt\m\al_3\ldots\al_{2N}}q^\m_{\tau\tau}q^{\al_3}_\tau\cdots
q^{\al_{2N}}_\tau  + G^{1-1/2N}F_{\bt\m}q^\m_\tau,\nonumber\\
&& F_{\la\m}=\dr_\la A_\m-\dr_\m A_\la. \nonumber
\een
It is readily observed that the variational derivatives $\cE_\la$
(\ref{kk2}) satisfy the Noether identities (\ref{s60'}). Moreover,
any relativistic Lagrangian obeying the Noether identity
(\ref{s60'}) is of type (\ref{kk1}).

A glance at the Lagrange equations (\ref{kk2}) shows that they
hold if
\mar{kk5}\beq
E_\bt= \F G_{\bt\nu_2\ldots\nu_{2N}}q^{\nu_2}_\tau\cdots
q^{\nu_{2N}}_\tau G^{-1}, \label{kk5}
\eeq
where $\F$ is some function on $J^1Q_R$. In particular, we
consider the equations
\mar{kk4}\beq
E_\bt=0. \label{kk4}
\eeq

Because of the Noether identities (\ref{s60'}), the system of
equations (\ref{kk2}) is underdetermined. To overcome this
difficulty, one can complete it with some additional equation.
Given the function $G$ (\ref{kk1}), let us choose the condition
\mar{kk8}\beq
G=1. \label{kk8}
\eeq
Being positive, the function $G$ (\ref{kk1}) possesses a nowhere
vanishing differential. Therefore, its level surface $W_G$ defined
by the condition (\ref{kk8}) is a submanifold of $J^1Q_R$.

Our choice of the equations (\ref{kk4}) and the condition
(\ref{kk8}) is motivated by the following facts.

\begin{lem} \label{kk12} \mar{kk12} Any solution of the
Lagrange equations (\ref{kk2}) living in the submanifold $W_G$ is
a solution of the equation (\ref{kk4}).
\end{lem}

\begin{proof}
A solution of the Lagrange equations (\ref{kk2}) living in the
submanifold $W_G$ obeys the system of equations
\mar{kk13}\beq
\cE_\la=0, \qquad G=1. \label{kk13}
\eeq
Therefore, it satisfies the equality
\mar{kk16}\beq
d_\tau G=0. \label{kk16}
\eeq
Then a glance at the expression (\ref{kk2}) shows that the
equations (\ref{kk13}) are equivalent to the equations
\mar{kk14}\ben
&& E_\la= \left(\frac{\dr_\la
G_{\m\al_2\ldots\al_{2N}}}{2N}- \dr_\m
G_{\la\al_2\ldots\al_{2N}}\right) q^\m_\tau q^{\al_2}_\tau\cdots
q^{\al_{2N}}_\tau - \nonumber\\
&& \qquad (2N-1)G_{\bt\m\al_3\ldots\al_{2N}}q^\m_{\tau\tau}q^{\al_3}_\tau\cdots
q^{\al_{2N}}_\tau  + F_{\bt\m}q^\m_\tau =0, \label{kk14}\\
&& G=G_{\al_1\ldots\al_{2N}}q^{\al_1}_\tau\cdots q^{\al_{2N}}_\tau=1.
\nonumber
\een
\end{proof}

\begin{lem} \label{kk12'} \mar{kk12'}
Solutions of the equations (\ref{kk4}) do not leave the
submanifold $W_G$ (\ref{kk8}).
\end{lem}

\begin{proof}
Since
\be
d_\tau G= -\frac{2N}{2N-1}q^\bt_\tau E_\bt,
\ee
any solution of the equations (\ref{kk4}) intersecting the
submanifold $W_G$ (\ref{kk8}) obeys the equality (\ref{kk16}) and,
consequently, lives in $W_G$.
\end{proof}

The system of equations (\ref{kk14}) is called the relativistic
equation. Its components $E_\la$ (\ref{kk3}) are not independent,
but obeys the relation
\be
q^\bt_\tau E_\bt=- \frac{2N-1}{2N}d_\tau G=0, \qquad G=1,
\ee
similar to the Noether identities (\ref{s60'}). The condition
(\ref{kk8}) is called the relativistic constraint.

Though the system of equations (\ref{kk2}) for sections of a fibre
bundle $Q_R\to\Bbb R$ is underdetermined, it is determined if,
given a coordinate chart $(U;q^0,q^i)$ (\ref{0303}) of $Q$ and the
corresponding coordinate chart (\ref{s20a}) of $Q_R$, we rewrite
it in the terms of three-velocities $q^i_0$ (\ref{s31'}) as
equations for sections of a fibre bundle $U\to\chi(U)\subset\Bbb
R$.

Let us denote
\mar{kk60}\beq
\ol G(q^\la,q^i_0)=(q^0_\tau)^{-2N}G(q^\la,q^\la_\tau), \qquad
q^0_\tau\neq 0. \label{kk60}
\eeq
Then we have
\be
\cE_i=q^0_\tau\left[\frac{\dr_i \ol G}{2N\ol G^{1-1/2N}} -
(q^0_\tau)^{-1}d_\tau\left(\frac{\dr_i^0 \ol G}{2N\ol
G^{1-1/2N}}\right)+ F_{ij}q^j_0 + F_{i0}\right].
\ee
Let us consider a solution $\{s^\la(\tau)\}$ of the equations
(\ref{kk2}) such that $\dr_\tau s^0$ does not vanish and there
exists an inverse function $\tau(q^0)$. Then this solution can be
represented by sections
\mar{kk51}\beq
s^i(\tau)=(\ol s^i\circ s^0)(\tau) \label{kk51}
\eeq
of the composite bundle
\be
\Bbb R\times U\to \Bbb R\times \p(U)\to \Bbb R
\ee
where $\ol s^i(q^0)=s^i(\tau(q^0))$ are sections of $U\to \chi(U)$
and $s^0(\tau)$ are sections of $\Bbb R\times \p(U)\to \Bbb R$.
Restricted to such solutions, the equations (\ref{kk2}) are
equivalent to the equations
\mar{kk20}\ben
&& \ol\cE_i=\frac{\dr_i \ol G}{2N\ol G^{1-1/2N}} -
d_0\left(\frac{\dr_i^0 \ol G}{2N\ol G^{1-1/2N}}\right)+
F_{ij}q^j_0 + F_{i0}=0,
\label{kk20} \\
&& \ol\cE_0=-q^i_0\ol\cE_i. \nonumber
\een
for sections $\ol s^i(q^0)$ of a fibre bundle $U\to\chi(U)$.

It is readily observed that the equations (\ref{kk20}) are the
Lagrange equation of the Lagrangian
\mar{kk21}\beq
\ol L=(\ol G^{1/2N} + q^i_0 A_i + A_0)dq^0 \label{kk21}
\eeq
on the jet manifold $J^1U$ of a fibre bundle $U\to\chi(U)$.

It should be emphasized that, both the equations (\ref{kk20}) and
the Lagrangian (\ref{kk21}) are defined only on a coordinate chart
(\ref{0303}) of $Q$ since they are not maintained under the
transition functions (\ref{s120}) -- (\ref{s120'}).

A solution $\ol s^i(q^0)$ of the equations (\ref{kk20}) defines a
solution $s^\la(\tau)$ (\ref{kk51}) of the equations (\ref{kk2})
up to an arbitrary function $s^0(\tau)$. The relativistic
constraint (\ref{kk8}) enables one to overcome this ambiguity as
follows.

Let us assume that, restricted to the coordinate chart
$(U;q^0,q^i)$ (\ref{0303}) of $Q$, the relativistic constraint
(\ref{kk8}) has no solution $q^0_\tau=0$. Then it is brought into
the form
\mar{kk62}\beq
(q^0_\tau)^{2N}\ol G(q^\la,q^i_0)=1, \label{kk62}
\eeq
where $\ol G$ is the function (\ref{kk60}). With the condition
(\ref{kk62}), every three-velocity $(q^i_0)$ defines a unique pair
of four-velocities
\mar{kk63}\beq
q^0_\tau = \pm (\ol G(q^\la,q^i_0))^{1/2N}, \qquad
q^i_\tau=q_\tau^0q^i_0. \label{kk63}
\eeq
Accordingly, any solution $\ol s^i(q^0)$ of the equations
(\ref{kk20}) leads to solutions
\be
\tau(q^0)=\pm\int (\ol G(q^0,\ol s^i(q^0),\dr_0\ol
s^i(q_0))^{-1/2N}dq^0, \quad s^i(\tau)=s^0(\tau)(\dr_i\ol
s^i)(s^0(\tau))
\ee
of the equations (\ref{kk13}) and, equivalently, the relativistic
equations (\ref{kk14}).

\begin{ex} \label{0313} \mar{0313}
Let $Q=M^4$ be a Minkowski space provided with the Minkowski
metric $\eta_{\m\nu}$ of signature $(+,---)$. This is the case of
Special Relativity. Let $\cA=\cA_\la dq^\la$ be a one-form on $Q$.
Then
\mar{s133}\beq
L=[m(\eta_{\m\nu}q^\m_\tau q^\nu_\tau)^{1/2} +e \cA_\m
q^\m_\tau]d\tau, \qquad m,e\in\Bbb R,\label{s133}
\eeq
is a relativistic Lagrangian on $J^1Q_R$ which satisfies the
Noether identity (\ref{s60'}). The corresponding relativistic
equation (\ref{kk14}) reads
\mar{x4,'}\ben
&& m\eta_{\mu\nu}q^\nu_{\tau\tau} -eF_{\m\nu}q^\nu_\tau=0,
\label{x4}\\
&& \eta_{\m\nu}q^\m_\tau q^\nu_\tau=1.  \label{x4'}
\een
This describes a relativistic massive charge in the presence of an
electromagnetic or gauge potential $\cA$. It follows from the
relativistic constraint (\ref{x4'}) that $(q^0_\tau)^2\geq 1$.
Therefore, passing to three-velocities, we obtain the Lagrangian
(\ref{kk21}):
\be
\ol L=\left[m(1-\op\sum_i (q^i_0)^2)^{1/2} +e (\cA_i q^i_0
+\cA_0)\right]dq^0,
\ee
and the Lagrange equations (\ref{kk20}):
\be
d_0\left(\frac{mq^i_0}{(1-\op\sum_i (q^i_0)^2)^{1/2}}\right)
+e(F_{ij}q^j_0 + F_{i0})=0.
\ee
\end{ex}

\end{document}